\documentclass[oneside,reqno,english]{amsart}
\usepackage[T1]{fontenc}
\usepackage[utf8]{inputenc}
\usepackage{xcolor}
\usepackage{babel}
\usepackage{prettyref}
\usepackage{amstext}
\usepackage{amsthm}
\usepackage{amssymb}
\usepackage[pdfusetitle,
 bookmarks=true,bookmarksnumbered=false,bookmarksopen=false,
 breaklinks=false,pdfborder={0 0 0},pdfborderstyle={},backref=false,colorlinks=false]
 {hyperref}
\hypersetup{
 colorlinks=true,citecolor=blue,linkcolor=blue,linktocpage=true}

\makeatletter
\numberwithin{equation}{section}
\numberwithin{figure}{section}

\usepackage{prettyref}

\newrefformat{cor}{Corollary~\ref{#1}}
\newrefformat{subsec}{Section~\ref{#1}}
\newrefformat{lem}{Lemma~\ref{#1}}
\newrefformat{thm}{Theorem~\ref{#1}}
\newrefformat{sec}{Section~\ref{#1}}
\newrefformat{chap}{Chapter~\ref{#1}}
\newrefformat{prop}{Proposition~\ref{#1}}
\newrefformat{exa}{Example~\ref{#1}}
\newrefformat{tab}{Table~\ref{#1}}
\newrefformat{rem}{Remark~\ref{#1}}
\newrefformat{def}{Definition~\ref{#1}}
\newrefformat{fig}{Figure~\ref{#1}}
\newrefformat{claim}{Claim~\ref{#1}}
\newrefformat{assu}{Assumption~\ref{#1}}

\makeatother

\theoremstyle{plain}
\newtheorem{thm}{\protect\theoremname}[section]
\newtheorem{prop}[thm]{\protect\propositionname}
\theoremstyle{definition}
\newtheorem{example}[thm]{\protect\examplename}
\theoremstyle{remark}
\newtheorem{rem}[thm]{\protect\remarkname}
\theoremstyle{plain}
\newtheorem{cor}[thm]{\protect\corollaryname}
\providecommand{\corollaryname}{Corollary}
\providecommand{\examplename}{Example}
\providecommand{\propositionname}{Proposition}
\providecommand{\remarkname}{Remark}
\providecommand{\theoremname}{Theorem}

\begin{document}
\title[Multipartite parity bounds and total correlation]{Multipartite parity bounds and total correlation}
\begin{abstract}
This paper studies multipartite observables formed from sums of local
self-adjoint contractions on tensor product Hilbert spaces. The square
of such a sum has a parity structure: after decomposing each local
product into commutator and anticommutator parts, the odd parity terms
cancel and only even parity contributions remain. This yields a norm
bound in terms of a family of pairwise defect weights built from local
commutator and anticommutator norms. These defect weights also control
an information theoretic estimate. The excess of the observable expectation
above the product state threshold is shown to necessarily carry a
definite amount of total correlation. Under a natural $\ell^{2}$-type
bound on each local family, this product state threshold becomes explicit,
which leads to a fully explicit lower bound on total correlation.
A simple depolarizing example illustrates the resulting decay mechanism
under local noise. 
\end{abstract}

\author{James Tian}
\address{Mathematical Reviews, 535 W. William St, Suite 210, Ann Arbor, MI
48103, USA}
\email{james.ftian@gmail.com}

\maketitle
\tableofcontents{}

\section{Introduction}\label{sec:1}

Let $H_{1},\dots,H_{n}$ be complex Hilbert spaces and let 
\[
a^{\left(r\right)}_{1},\dots,a^{\left(r\right)}_{m}\in B\left(H_{r}\right),\qquad r\in\left\{ 1,\dots,n\right\} ,
\]
be self-adjoint contractions. We study the multipartite sum 
\[
B=\sum^{m}_{i=1}a^{\left(1\right)}_{i}\otimes\cdots\otimes a^{\left(n\right)}_{i}\in B\left(H_{1}\otimes\cdots\otimes H_{n}\right).
\]
Operators of this form arise naturally as Bell type operators built
from local observables, and more generally as tensor-valued analogues
of classical signed sums. Two questions then present themselves. First,
how can one control $\left\Vert B\right\Vert $ in terms of the local
commutator and anticommutator structure? Second, when a state $\rho$
satisfies
\[
\mathrm{Tr}\left(\rho B\right)>\mathrm{Tr}\left(\sigma B\right)
\]
for every product state $\sigma$, to what extent does this already
quantify the total correlation 
\[
I_{\mathrm{tot}}(\rho):=D\left(\rho\Vert\rho_{1}\otimes\cdots\otimes\rho_{n}\right),
\]
where $\rho_{r}$ denotes the marginal of $\rho$ on $H_{r}$?

The background for these questions comes from two related lines of
work. On one side, Bell inequalities, multipartite witnesses, and
related positivity questions have been studied from the viewpoints
of optimization, operator inequalities, and quantum information; see
for example \cite{MR4783206,MR4496629,MR4411389}. On another side,
relative entropy, tensorization, and quantitative estimates for quantum
Markov evolutions have led to a growing body of information theoretic
and operator algebraic inequalities; see for example \cite{MR4361872,MR4444072,MR4518485,MR4860032}.
The present paper sits somewhat differently in this picture. We do
not try to classify Bell inequalities or optimize a fixed family of
multipartite correlation inequalities. Instead, we isolate an operator
theoretic mechanism in the square of $B$ itself and use it to pass
from local commutator and anticommutator data to explicit lower bounds
on total correlation.

The starting point is that $B^{2}$ carries a built-in parity structure.
After expanding $B^{2}$, the mixed terms do not lead to an arbitrary
family of tensor expressions. At each site, the local product $a^{\left(r\right)}_{i}a^{\left(r\right)}_{j}$
splits into its commutator and anticommutator parts. When these local
pieces are tensored together, the odd parity contributions cancel,
and only the even parity sector survives. This produces a canonical
family of defect weights $\phi^{\left(n\right)}_{ij}$ indexed by
unordered pairs $\left\{ i,j\right\} $, and leads to the norm estimate
\[
\left\Vert B\right\Vert ^{2}\leq m+\sum_{1\leq i<j\leq m}\phi^{\left(n\right)}_{ij}
\]
proved in \prettyref{prop:2-1}. This is the basic structural result
of the paper.

The same defect quantity then enters the information theoretic part
of the argument. Let $\mathcal{P}$ denote the set of product states
on $H_{1}\otimes\cdots\otimes H_{n}$ and define the product threshold
\[
\Gamma_{\mathrm{prod}}\left(B\right):=\sup_{\sigma\in\mathcal{P}}\mathrm{Tr}\left(\sigma B\right).
\]
For a state $\rho$, the excess 
\[
\Delta_{B}\left(\rho\right):=\left(\mathrm{Tr}\left(\rho B\right)-\Gamma_{\mathrm{prod}}\left(B\right)\right)_{+}
\]
measures the amount by which the value of $B$ on $\rho$ exceeds
the product threshold. \prettyref{thm:3-1} shows that a positive
value of $\Delta_{B}\left(\rho\right)$ cannot occur for free: it
entails both separation from the product class in trace norm and a
definite amount of total correlation. In this way, the same denominator
that controls $\left\Vert B\right\Vert $ also controls the information
cost of exceeding the product threshold. This part of the argument
is close in spirit to recent work on quasi-factorization, tensorization,
and entropic inequalities in noncommutative settings; see for example
\cite{MR4361872,MR4518485,MR4444072,MR4520657}.

To make this useful, one still needs to control the threshold itself.
Our next step is therefore to impose a simple quadratic bound on the
expectation vectors associated with each local family $\{a^{\left(r\right)}_{i}\}$.
Under this hypothesis, \prettyref{thm:3-4} gives an explicit upper
bound for $\Gamma_{\mathrm{prod}}\left(B\right)$. Combined with \prettyref{thm:3-1},
this yields the explicit correlation estimate 
\[
I_{\mathrm{tot}}\left(\rho\right)\geq\frac{1}{2}\frac{\left(\mathrm{Tr}\left(\rho B\right)-\prod^{n}_{r=1}C^{1/2}_{r}\right)^{2}_{+}}{m+\sum_{1\leq i<j\leq m}\phi^{\left(n\right)}_{ij}}
\]
in \prettyref{cor:3-5}. In this form the roles of the two structural
quantities become clear. The numerator measures excess above an explicit
product threshold, while the denominator measures the mixed commutator/anticommutator
complexity of $B$.

We also include a simple dynamical illustration. For local depolarizing
semigroups, centered observables are eigenvectors for the Heisenberg
evolution, so the multipartite observable $B$ decays by an exact
scalar factor. This gives a concrete model in which the excess $\Delta_{B}\left(\rho_{t}\right)$
over the product threshold and the corresponding lower bound on total
correlation can be followed explicitly along a natural local noise
flow. We include this example mainly to show that the parity defect
denominator remains visible beyond the static setting. For related
recent work on entropy decay, decoherence rates, and geometric features
of quantum channels and quantum Markov semigroups, see for example
\cite{MR4282395,MR4496596,MR4703456,MR4704528}.

From this point of view, the contribution of the paper is to isolate
the parity expansion for $B^{2}$ and the resulting defect weights
$\phi^{\left(n\right)}_{ij}$, and then to derive the correlation
bounds that follow from them. Once that expansion is in place, one
can pass from operator control to information theoretic control, and
under a uniform $\ell^{2}$ bound on the local expectation vectors
one obtains a fully explicit lower bound on total correlation.

The paper is organized accordingly: \prettyref{sec:2} develops the
parity expansion and the norm bound, while \prettyref{sec:3} derives
the resulting static correlation estimates, including the product
threshold bound and the explicit lower bound on total correlation.
\prettyref{sec:4} then applies these estimates to product quantum
Markov semigroups and local depolarizing noise. The final section
contains some brief concluding remarks.

\section{Multipartite Parity Bounds}\label{sec:2}

We begin with the basic multipartite norm estimate. The point is that,
after expanding $B^{2}$, the mixed terms admit a canonical even parity
decomposition across the tensor factors. This produces a natural family
of pair weights $\phi^{\left(n\right)}_{ij}$ and reduces the operator
bound to a finite sum of commutator and anticommutator defects. In
the bipartite case $n=2$ this recovers the complete-graph inequality
from \cite{tian2025ineq}, while here we work in the general multipartite
setting.
\begin{prop}
\label{prop:2-1} For each $r\in\left\{ 1,\dots,n\right\} $, let
$H_{r}$ be a complex Hilbert space. Let 
\[
a^{\left(r\right)}_{1},\dots,a^{\left(r\right)}_{m}\in B\left(H_{r}\right)
\]
be self-adjoint contractions. Set 
\[
\mathcal{H}:=H_{1}\otimes\cdots\otimes H_{n}
\]
and define 
\[
B:=\sum^{m}_{i=1}a^{\left(1\right)}_{i}\otimes\cdots\otimes a^{\left(n\right)}_{i}\in B\left(\mathcal{H}\right).
\]
For $1\leq i<j\leq m$, define 
\[
\phi^{\left(n\right)}_{ij}:=2^{1-n}\sum_{\substack{S\subseteq\left\{ 1,\dots,n\right\} \\
\left|S\right|\text{ even}
}
}\prod_{r\in S}\left\Vert \left[a^{\left(r\right)}_{i},a^{\left(r\right)}_{j}\right]\right\Vert \prod_{r\notin S}\left\Vert \left\{ a^{\left(r\right)}_{i},a^{\left(r\right)}_{j}\right\} \right\Vert .
\]
Then 
\begin{equation}
\left\Vert B\right\Vert ^{2}\leq m+\sum_{1\leq i<j\leq m}\phi^{\left(n\right)}_{ij}.\label{eq:2-1}
\end{equation}
\end{prop}

\begin{proof}
For each $i\in\left\{ 1,\dots,m\right\} $, set 
\[
u_{i}:=a^{\left(1\right)}_{i}\otimes\cdots\otimes a^{\left(n\right)}_{i}\in B\left(\mathcal{H}\right).
\]
Since each $a^{\left(r\right)}_{i}$ is self-adjoint, $u_{i}$ is
self-adjoint. Since each $a^{\left(r\right)}_{i}$ is a contraction,
we also have 
\[
\left\Vert u_{i}\right\Vert =\prod^{n}_{r=1}\Vert a^{\left(r\right)}_{i}\Vert\leq1.
\]
Hence $B=\sum^{m}_{i=1}u_{i}$ and 
\[
B^{2}=\sum^{m}_{i=1}u^{2}_{i}+\sum_{1\leq i<j\leq m}\left(u_{i}u_{j}+u_{j}u_{i}\right).
\]
Taking norms and using the triangle inequality gives 
\[
\left\Vert B\right\Vert ^{2}=\left\Vert B^{2}\right\Vert \leq\left\Vert \sum\nolimits^{m}_{i=1}u^{2}_{i}\right\Vert +\sum\nolimits_{1\leq i<j\leq m}\left\Vert u_{i}u_{j}+u_{j}u_{i}\right\Vert .
\]
Since $\left\Vert u^{2}_{i}\right\Vert \leq\left\Vert u_{i}\right\Vert ^{2}\leq1$,
we obtain 
\[
\left\Vert \sum\nolimits^{m}_{i=1}u^{2}_{i}\right\Vert \leq\sum^{m}_{i=1}\left\Vert u^{2}_{i}\right\Vert \leq m.
\]
It remains to estimate the mixed terms.

Fix $1\leq i<j\leq m$. For each $r\in\left\{ 1,\dots,n\right\} $,
write 
\[
A^{\left(r\right)}_{ij}:=\left\{ a^{\left(r\right)}_{i},a^{\left(r\right)}_{j}\right\} ,\qquad C^{\left(r\right)}_{ij}:=\left[a^{\left(r\right)}_{i},a^{\left(r\right)}_{j}\right].
\]
Then 
\[
a^{\left(r\right)}_{i}a^{\left(r\right)}_{j}=\frac{1}{2}\left(A^{\left(r\right)}_{ij}+C^{\left(r\right)}_{ij}\right),\qquad a^{\left(r\right)}_{j}a^{\left(r\right)}_{i}=\frac{1}{2}\left(A^{\left(r\right)}_{ij}-C^{\left(r\right)}_{ij}\right).
\]
Hence 
\[
u_{i}u_{j}=2^{-n}\bigotimes^{n}_{r=1}\left(A^{\left(r\right)}_{ij}+C^{\left(r\right)}_{ij}\right),
\]
and 
\[
u_{j}u_{i}=2^{-n}\bigotimes^{n}_{r=1}\left(A^{\left(r\right)}_{ij}-C^{\left(r\right)}_{ij}\right).
\]

We now expand both tensor products. For each subset $S\subseteq\left\{ 1,\dots,n\right\} $,
let 
\[
T_{S}:=\bigotimes^{n}_{r=1}X^{\left(S\right)}_{r},
\]
where 
\[
X^{\left(S\right)}_{r}=\begin{cases}
C^{\left(r\right)}_{ij}, & r\in S,\\
A^{\left(r\right)}_{ij}, & r\notin S.
\end{cases}
\]
Then 
\[
u_{i}u_{j}=2^{-n}\sum_{S\subseteq\left\{ 1,\dots,n\right\} }T_{S},
\]
while 
\[
u_{j}u_{i}=2^{-n}\sum_{S\subseteq\left\{ 1,\dots,n\right\} }(-1)^{\left|S\right|}T_{S}.
\]
Therefore 
\[
u_{i}u_{j}+u_{j}u_{i}=2^{-n}\sum_{S\subseteq\left\{ 1,\dots,n\right\} }\left(1+(-1)^{\left|S\right|}\right)T_{S}.
\]
If $\left|S\right|$ is odd, then $1+(-1)^{\left|S\right|}=0$. If
$\left|S\right|$ is even, then $1+(-1)^{\left|S\right|}=2$. Thus
\[
u_{i}u_{j}+u_{j}u_{i}=2^{1-n}\sum_{\substack{S\subseteq\left\{ 1,\dots,n\right\} \\
\left|S\right|\text{ even}
}
}T_{S}.
\]
Taking norms and using the triangle inequality yields 
\[
\left\Vert u_{i}u_{j}+u_{j}u_{i}\right\Vert \leq2^{1-n}\sum_{\substack{S\subseteq\left\{ 1,\dots,n\right\} \\
\left|S\right|\text{ even}
}
}\left\Vert T_{S}\right\Vert .
\]
It follows that
\[
\left\Vert T_{S}\right\Vert =\prod_{r\in S}\Vert C^{\left(r\right)}_{ij}\Vert\prod_{r\notin S}\Vert A^{\left(r\right)}_{ij}\Vert.
\]
Substituting the definitions of $A^{\left(r\right)}_{ij}$ and $C^{\left(r\right)}_{ij}$,
we obtain 
\[
\left\Vert u_{i}u_{j}+u_{j}u_{i}\right\Vert \leq\phi^{\left(n\right)}_{ij}.
\]

Returning to the expansion of $B^{2}$, we get \prettyref{eq:2-1}.
\end{proof}

We now give a simple three-qubit Pauli-string example showing that
the estimate in \prettyref{eq:2-1} can be sharp already in a genuinely
tripartite setting.
\begin{example}
Let $n=3$, let $H_{1}=H_{2}=H_{3}=\mathbb{C}^{2}$, and let 
\[
X=\begin{pmatrix}0 & 1\\
1 & 0
\end{pmatrix},\qquad Y=\begin{pmatrix}0 & -i\\
i & 0
\end{pmatrix},\qquad Z=\begin{pmatrix}1 & 0\\
0 & -1
\end{pmatrix}
\]
be the Pauli matrices. Define 
\[
u_{1}:=X\otimes Y\otimes X,\qquad u_{2}:=X\otimes Y\otimes Z,\qquad u_{3}:=Z\otimes Z\otimes Z,
\]
and set 
\[
B:=u_{1}+u_{2}+u_{3}.
\]
Each $u_{i}$ is a self-adjoint unitary, hence $\left\Vert u_{i}\right\Vert =1$.

We first compute $\left\Vert B\right\Vert ^{2}$. Since $u^{2}_{i}=I$
for each $i$, we have 
\[
B^{2}=3I+\left(u_{1}u_{2}+u_{2}u_{1}\right)+\left(u_{1}u_{3}+u_{3}u_{1}\right)+\left(u_{2}u_{3}+u_{3}u_{2}\right).
\]
Now 
\[
u_{1}u_{2}=\left(XX\right)\otimes\left(YY\right)\otimes\left(XZ\right)=I\otimes I\otimes\left(XZ\right),
\]
while 
\[
u_{2}u_{1}=\left(XX\right)\otimes\left(YY\right)\otimes\left(ZX\right)=I\otimes I\otimes\left(ZX\right).
\]
Since $XZ=-ZX$, it follows that 
\[
u_{1}u_{2}+u_{2}u_{1}=0.
\]
Similarly, 
\[
u_{1}u_{3}=\left(XZ\right)\otimes\left(YZ\right)\otimes\left(XZ\right),\qquad u_{3}u_{1}=\left(ZX\right)\otimes\left(ZY\right)\otimes\left(ZX\right),
\]
and since $X,Z$ anticommute and $Y,Z$ anticommute, the three local
sign changes give 
\[
u_{1}u_{3}=-u_{3}u_{1},
\]
hence 
\[
u_{1}u_{3}+u_{3}u_{1}=0.
\]
For the remaining pair, we obtain 
\[
u_{2}u_{3}=\left(XZ\right)\otimes\left(YZ\right)\otimes\left(ZZ\right),\qquad u_{3}u_{2}=\left(ZX\right)\otimes\left(ZY\right)\otimes\left(ZZ\right).
\]
There are exactly two local sign changes, so $u_{2}$ and $u_{3}$
commute: 
\[
u_{2}u_{3}=u_{3}u_{2}.
\]
Therefore 
\[
B^{2}=3I+2u_{2}u_{3}.
\]
Since $u_{2}u_{3}$ is again a self-adjoint unitary, its spectrum
is contained in $\left\{ -1,1\right\} $. It follows that 
\[
\mathrm{spec}\left(B^{2}\right)=\left\{ 1,5\right\} ,\qquad\left\Vert B\right\Vert ^{2}=5.
\]

We now compute the defect weights in \prettyref{prop:2-1}. For Pauli
matrices, the local commutator-anticommutator norms are 
\[
\left\Vert \left\{ P,P\right\} \right\Vert =2,\qquad\left\Vert i\left[P,P\right]\right\Vert =0,
\]
and for distinct Pauli matrices $P\neq Q$, 
\[
\left\Vert \left\{ P,Q\right\} \right\Vert =0,\qquad\left\Vert i\left[P,Q\right]\right\Vert =2.
\]

For the pair $(1,2)$, the local pairs are 
\[
\left(a^{(1)}_{1},a^{(1)}_{2}\right)=\left(X,X\right),\qquad\left(a^{(2)}_{1},a^{(2)}_{2}\right)=\left(Y,Y\right),\qquad\left(a^{(3)}_{1},a^{(3)}_{2}\right)=\left(X,Z\right).
\]
Hence 
\[
\left\Vert \left\{ a^{(1)}_{1},a^{(1)}_{2}\right\} \right\Vert =2,\qquad\left\Vert \left\{ a^{(2)}_{1},a^{(2)}_{2}\right\} \right\Vert =2,\qquad\left\Vert \left\{ a^{(3)}_{1},a^{(3)}_{2}\right\} \right\Vert =0,
\]
and 
\[
\left\Vert i\left[a^{(1)}_{1},a^{(1)}_{2}\right]\right\Vert =0,\qquad\left\Vert i\left[a^{(2)}_{1},a^{(2)}_{2}\right]\right\Vert =0,\qquad\left\Vert i\left[a^{(3)}_{1},a^{(3)}_{2}\right]\right\Vert =2.
\]
Since $n=3$, the even subsets of $\left\{ 1,2,3\right\} $ are 
\[
\varnothing,\qquad\left\{ 1,2\right\} ,\qquad\left\{ 1,3\right\} ,\qquad\left\{ 2,3\right\} .
\]
Thus every term in the defining sum for $\phi^{(3)}_{12}$ contains
at least one vanishing factor: 
\[
\phi^{(3)}_{12}=2^{-2}\left[\left(2\cdot2\cdot0\right)+\left(0\cdot0\cdot0\right)+\left(0\cdot2\cdot2\right)+\left(2\cdot0\cdot2\right)\right]=0.
\]

For the pair $(1,3)$, the local pairs are 
\[
\left(X,Z\right),\qquad\left(Y,Z\right),\qquad\left(X,Z\right),
\]
so all anticommutator norms vanish and all commutator norms are equal
to $2$. Again every even-parity term contains at least one anticommutator
factor, hence 
\[
\phi^{(3)}_{13}=0.
\]

For the pair $(2,3)$, the local pairs are 
\[
\left(X,Z\right),\qquad\left(Y,Z\right),\qquad\left(Z,Z\right).
\]
Therefore 
\[
\left\Vert \left\{ a^{(1)}_{2},a^{(1)}_{3}\right\} \right\Vert =0,\qquad\left\Vert \left\{ a^{(2)}_{2},a^{(2)}_{3}\right\} \right\Vert =0,\qquad\left\Vert \left\{ a^{(3)}_{2},a^{(3)}_{3}\right\} \right\Vert =2,
\]
and 
\[
\left\Vert i\left[a^{(1)}_{2},a^{(1)}_{3}\right]\right\Vert =2,\qquad\left\Vert i\left[a^{(2)}_{2},a^{(2)}_{3}\right]\right\Vert =2,\qquad\left\Vert i\left[a^{(3)}_{2},a^{(3)}_{3}\right]\right\Vert =0.
\]
Hence 
\[
\phi^{(3)}_{23}=2^{-2}\left[\left(0\cdot0\cdot2\right)+\left(2\cdot2\cdot2\right)+\left(2\cdot2\cdot0\right)+\left(0\cdot2\cdot0\right)\right]=2.
\]

We conclude that 
\[
\sum_{1\leq i<j\leq3}\phi^{(3)}_{ij}=0+0+2=2.
\]
Since $m=3$, \prettyref{prop:2-1} yields 
\[
\left\Vert B\right\Vert ^{2}\leq m+\sum_{1\leq i<j\leq3}\phi^{(3)}_{ij}=3+2=5.
\]
Combining this with the exact computation above, we obtain 
\[
\left\Vert B\right\Vert ^{2}=5=m+\sum_{1\leq i<j\leq3}\phi^{(3)}_{ij}.
\]
Thus the estimate in \prettyref{prop:2-1} is sharp in this tripartite
example. 
\end{example}

\section{Static Correlation Bounds}\label{sec:3}

The norm estimate in \prettyref{prop:2-1} has an immediate information
theoretic consequence. Once $\left\Vert B\right\Vert $ is controlled
in terms of the parity defects $\phi^{\left(n\right)}_{ij}$, any
state whose expectation of $B$ exceeds the product threshold must
lie quantitatively far from the set of product states. In particular,
a positive value of $\Delta_{B}\left(\rho\right)$ as in \prettyref{eq:2-4}
ensures a nonzero lower bound for total correlation.
\begin{thm}
\label{thm:3-1} In the setting of \prettyref{prop:2-1}, let 
\begin{equation}
\mathcal{P}:=\left\{ \sigma^{\left(1\right)}\otimes\cdots\otimes\sigma^{\left(n\right)}:\sigma^{\left(r\right)}\ \text{is a state on }H_{r}\ \text{for each }r\right\} \label{eq:2-2}
\end{equation}
denote the set of product states on $\mathcal{H}$. Define the product
threshold 
\begin{equation}
\Gamma_{\mathrm{prod}}\left(B\right):=\sup_{\sigma\in\mathcal{P}}\mathrm{Tr}\left(\sigma B\right),\label{eq:2-3}
\end{equation}
and, for a state $\rho$ on $\mathcal{H}$, the excess
\begin{equation}
\Delta_{B}\left(\rho\right):=\left(\mathrm{Tr}\left(\rho B\right)-\Gamma_{\mathrm{prod}}\left(B\right)\right)_{+}.\label{eq:2-4}
\end{equation}
Then the following hold. 
\begin{enumerate}
\item For every state $\rho$ on $\mathcal{H}$, 
\begin{equation}
\inf_{\sigma\in\mathcal{P}}\left\Vert \rho-\sigma\right\Vert _{1}\geq\frac{\Delta_{B}\left(\rho\right)}{\left(m+\sum_{1\leq i<j\leq m}\phi^{\left(n\right)}_{ij}\right)^{1/2}}.\label{eq:2-5}
\end{equation}
\item Let $\rho$ be a state on $\mathcal{H}$ and let $\rho_{r}$ denote
its marginal on $H_{r}$. Define the total correlation 
\[
I_{\mathrm{tot}}\left(\rho\right):=D\left(\rho\Vert\rho_{1}\otimes\cdots\otimes\rho_{n}\right),
\]
where $D\left(\cdot\middle\Vert\cdot\right)$ is the quantum relative
entropy with natural logarithm. Then 
\begin{equation}
I_{\mathrm{tot}}\left(\rho\right)\geq\frac{1}{2}\frac{\Delta_{B}\left(\rho\right)^{2}}{m+\sum_{1\leq i<j\leq m}\phi^{\left(n\right)}_{ij}}.\label{eq:2-6}
\end{equation}
\end{enumerate}
\end{thm}

\begin{proof}
We first prove (1). Fix a state $\rho$ on $\mathcal{H}$. If $\Delta_{B}\left(\rho\right)=0$
there is nothing to prove, so assume $\Delta_{B}\left(\rho\right)>0$.
For any $\sigma\in\mathcal{P}$ we have, by definition \prettyref{eq:2-3},
\[
\mathrm{Tr}\left(\rho B\right)-\Gamma_{\mathrm{prod}}\left(B\right)\leq\mathrm{Tr}\left(\rho B\right)-\mathrm{Tr}\left(\sigma B\right),
\]
and therefore 
\[
\Delta_{B}\left(\rho\right)=\left(\mathrm{Tr}\left(\rho B\right)-\Gamma_{\mathrm{prod}}\left(B\right)\right)_{+}\leq\left|\mathrm{Tr}\left(\rho B\right)-\mathrm{Tr}\left(\sigma B\right)\right|=\left|\mathrm{Tr}\left(\left(\rho-\sigma\right)B\right)\right|.
\]
Using the duality between trace norm and operator norm gives 
\[
\Delta_{B}\left(\rho\right)\leq\left\Vert B\right\Vert \left\Vert \rho-\sigma\right\Vert _{1}.
\]
Taking the infimum over $\sigma\in\mathcal{P}$ yields 
\[
\Delta_{B}\left(\rho\right)\leq\left\Vert B\right\Vert \inf_{\sigma\in\mathcal{P}}\left\Vert \rho-\sigma\right\Vert _{1}.
\]
Combining this with \prettyref{eq:2-1} gives \prettyref{eq:2-5}.

We now prove (2). Let $\rho$ be a state on $\mathcal{H}$ and let
$\rho_{r}$ be its marginal on $H_{r}$. Set 
\[
\sigma_{*}:=\rho_{1}\otimes\cdots\otimes\rho_{n}\in\mathcal{P}.
\]
Then 
\[
\inf_{\sigma\in\mathcal{P}}\left\Vert \rho-\sigma\right\Vert _{1}\leq\left\Vert \rho-\sigma_{*}\right\Vert _{1}.
\]
The quantum Pinsker inequality with natural logarithm states that
\[
D\left(\rho\Vert\sigma_{*}\right)\geq\frac{1}{2}\left\Vert \rho-\sigma_{*}\right\Vert ^{2}_{1}.
\]
Since 
\[
I_{\mathrm{tot}}\left(\rho\right)=D\left(\rho\middle\Vert\sigma_{*}\right),
\]
we obtain 
\[
I_{\mathrm{tot}}\left(\rho\right)\geq\frac{1}{2}\left\Vert \rho-\sigma_{*}\right\Vert ^{2}_{1}\geq\frac{1}{2}\left(\inf_{\sigma\in\mathcal{P}}\left\Vert \rho-\sigma\right\Vert _{1}\right)^{2}.
\]
Applying \prettyref{eq:2-5} to the right hand side gives 
\[
I_{\mathrm{tot}}\left(\rho\right)\geq\frac{1}{2}\frac{\Delta_{B}\left(\rho\right)^{2}}{m+\sum_{1\leq i<j\leq m}\phi^{\left(n\right)}_{ij}},
\]
which proves \prettyref{eq:2-6}.
\end{proof}

\begin{example}[CHSH total correlation bound]
\label{exa:3-2} We illustrate \prettyref{thm:3-1} in the bipartite
case $n=2$ using the standard CHSH operator. Let $H_{1}=H_{2}=\mathbb{C}^{2}$
and let $X,Z$ be the Pauli matrices 
\[
X=\begin{pmatrix}0 & 1\\
1 & 0
\end{pmatrix},\qquad Z=\begin{pmatrix}1 & 0\\
0 & -1
\end{pmatrix}.
\]
Define local observables 
\[
A_{1}:=Z,\qquad A_{2}:=X,\qquad B_{1}:=\frac{Z+X}{\sqrt{2}},\qquad B_{2}:=\frac{Z-X}{\sqrt{2}}.
\]
Each of $A_{1},A_{2},B_{1},B_{2}$ is a self-adjoint unitary, hence
a self-adjoint contraction.

We form the CHSH operator 
\[
\mathcal{B}:=A_{1}\otimes B_{1}+A_{1}\otimes B_{2}+A_{2}\otimes B_{1}-A_{2}\otimes B_{2}
\]
on $H_{1}\otimes H_{2}$. To place this in the form of \prettyref{prop:2-1},
write 
\[
\mathcal{B}=\sum^{4}_{i=1}a^{(1)}_{i}\otimes a^{(2)}_{i},
\]
with $a^{(1)}_{1}=A_{1}$, $a^{(2)}_{1}=B_{1}$, $a^{(1)}_{2}=A_{1}$,
$a^{(2)}_{2}=B_{2}$, $a^{(1)}_{3}=A_{2}$, $a^{(2)}_{3}=B_{1}$,
$a^{(1)}_{4}=A_{2}$, $a^{(2)}_{4}=-B_{2}$. Each $a^{(r)}_{i}$ is
a self-adjoint contraction, so the hypotheses of \prettyref{prop:2-1}
and \prettyref{thm:3-1} hold with $n=2$ and $m=4$ and $B=\mathcal{B}$.

\smallskip{}

\emph{(i) Product threshold.} Let $\mathcal{P}$ denote the set of
product states on $H_{1}\otimes H_{2}$ as in \prettyref{eq:2-2}.
For a product state $\sigma=\sigma^{(1)}\otimes\sigma^{(2)}$, let
$r=(r_{x},r_{y},r_{z})$ and $s=(s_{x},s_{y},s_{z})$ be the Bloch
vectors of $\sigma^{(1)}$ and $\sigma^{(2)}$, so that 
\begin{align*}
\mathrm{Tr}\left(\sigma^{(1)}X\right) & =r_{x},\quad\mathrm{Tr}\left(\sigma^{(1)}Z\right)=r_{z},\\
\mathrm{Tr}\left(\sigma^{(2)}X\right) & =s_{x},\quad\mathrm{Tr}\left(\sigma^{(2)}Z\right)=s_{z},
\end{align*}
with $r^{2}_{x}+r^{2}_{y}+r^{2}_{z}\leq1$ and $s^{2}_{x}+s^{2}_{y}+s^{2}_{z}\leq1$.
Then 
\[
\mathrm{Tr}\left(\sigma^{(2)}B_{1}\right)=\frac{s_{z}+s_{x}}{\sqrt{2}},\qquad\mathrm{Tr}\left(\sigma^{(2)}B_{2}\right)=\frac{s_{z}-s_{x}}{\sqrt{2}}.
\]
A direct computation gives 
\begin{align*}
\mathrm{Tr}\left(\sigma\mathcal{B}\right) & =\mathrm{Tr}\left(\sigma^{(1)}A_{1}\right)\mathrm{Tr}\left(\sigma^{(2)}\left(B_{1}+B_{2}\right)\right)+\mathrm{Tr}\left(\sigma^{(1)}A_{2}\right)\mathrm{Tr}\left(\sigma^{(2)}\left(B_{1}-B_{2}\right)\right)\\
 & =r_{z}\sqrt{2}s_{z}+r_{x}\sqrt{2}s_{x}=\sqrt{2}\left(r_{x}s_{x}+r_{z}s_{z}\right).
\end{align*}
Let $r'=(r_{x},r_{z})$ and $s'=(s_{x},s_{z})$. Since $\left|r'\right|\leq1$
and $\left|s'\right|\leq1$, the Cauchy-Schwarz inequality yields
\[
\left|\mathrm{Tr}\left(\sigma\mathcal{B}\right)\right|\leq\sqrt{2}\left|r'\cdot s'\right|\leq\sqrt{2}.
\]
Thus 
\[
\Gamma_{\mathrm{prod}}(\mathcal{B})=\sup_{\sigma\in\mathcal{P}}\mathrm{Tr}\left(\sigma\mathcal{B}\right)\leq\sqrt{2}.
\]
Equality holds, for example, when $r'$ and $s'$ are unit vectors
and parallel in the $x$--$z$ plane, so we have 
\[
\Gamma_{\mathrm{prod}}(\mathcal{B})=\sqrt{2}.
\]
Here $\Gamma_{\mathrm{prod}}(\mathcal{B})$ is the supremum over quantum
product states for this fixed choice of local observables.

\smallskip{}

\emph{(ii) A state with strictly positive excess.} It is well known
(and can be checked directly) that there exists a two-qubit state
$\rho_{*}$ such that 
\[
\mathrm{Tr}\left(\rho_{*}\mathcal{B}\right)=2\sqrt{2}
\]
and $\left\Vert \mathcal{B}\right\Vert =2\sqrt{2}$; one choice is
the maximally entangled Bell state $\rho_{*}=\left|\Phi^{+}\right\rangle \left\langle \Phi^{+}\right|$,
where 
\[
\left|\Phi^{+}\right\rangle =\frac{1}{\sqrt{2}}\left(\left|00\right\rangle +\left|11\right\rangle \right).
\]
For this state, 
\[
\Delta_{\mathcal{B}}(\rho_{*})=\left(\mathrm{Tr}\left(\rho_{*}\mathcal{B}\right)-\Gamma_{\mathrm{prod}}(\mathcal{B})\right)_{+}=\left(2\sqrt{2}-\sqrt{2}\right)_{+}=\sqrt{2}>0.
\]

\smallskip{}

\emph{(iii) Defect denominator.} In this example $n=2$, so for $1\leq i<j\leq4$
the defect weights from \prettyref{prop:2-1} take the form 
\[
\phi^{(2)}_{ij}=\frac{1}{2}\left\Vert \left\{ a^{(1)}_{i},a^{(1)}_{j}\right\} \right\Vert \left\Vert \left\{ a^{(2)}_{i},a^{(2)}_{j}\right\} \right\Vert +\frac{1}{2}\left\Vert \left[a^{(1)}_{i},a^{(1)}_{j}\right]\right\Vert \left\Vert \left[a^{(2)}_{i},a^{(2)}_{j}\right]\right\Vert .
\]
Using the Pauli relations 
\[
\left\{ P,P\right\} =2I,\quad\left[P,P\right]=0,\qquad\left\{ P,Q\right\} =0,\quad\left[P,Q\right]\neq0\ \text{for }P\neq Q,
\]
and the fact that $A_{1},A_{2},B_{1},B_{2}$ are linear combinations
of $X$ and $Z$, a short computation shows that 
\[
\left\{ B_{1},B_{2}\right\} =0,\qquad\left\Vert \left[B_{1},B_{2}\right]\right\Vert =2,
\]
and 
\[
\phi^{(2)}_{14}=\phi^{(2)}_{23}=2,\qquad\phi^{(2)}_{12}=\phi^{(2)}_{13}=\phi^{(2)}_{24}=\phi^{(2)}_{34}=0.
\]
Hence 
\[
\sum_{1\leq i<j\leq4}\phi^{(2)}_{ij}=4,\qquad m+\sum_{1\leq i<j\leq4}\phi^{(2)}_{ij}=4+4=8.
\]

\smallskip{}

\emph{(iv) Information-theoretic lower bound.} Let $\rho_{1}$ and
$\rho_{2}$ be the marginals of $\rho_{*}$ on $H_{1}$ and $H_{2}$,
and let 
\[
I_{\mathrm{tot}}(\rho_{*})=D\left(\rho_{*}\Vert\rho_{1}\otimes\rho_{2}\right)
\]
be the total correlation. Applying \prettyref{thm:3-1} with $B=\mathcal{B}$
and $\rho=\rho_{*}$ yields 
\[
I_{\mathrm{tot}}(\rho_{*})\geq\frac{1}{2}\frac{\Delta_{\mathcal{B}}(\rho_{*})^{2}}{m+\sum_{1\leq i<j\leq4}\phi^{(2)}_{ij}}=\frac{1}{2}\cdot\frac{\left(\sqrt{2}\right)^{2}}{8}=\frac{1}{8}.
\]
Thus, in this CHSH configuration, the fact that $\rho_{*}$ exceeds
the product threshold by $\sqrt{2}$ already forces a strictly positive,
explicit lower bound 
\[
I_{\mathrm{tot}}(\rho_{*})\geq\frac{1}{8}
\]
on its total correlation. 
\end{example}

\begin{rem}
The restriction to $n=2$ in \prettyref{exa:3-2} is only for explicitness.
For $n>2$ one may embed the same construction by replacing the bipartite
CHSH operator $B$ on $H_{1}\otimes H_{2}$ with 
\[
\widetilde{B}:=B\otimes I_{H_{3}}\otimes\cdots\otimes I_{H_{n}},
\]
and then considering states of the form 
\[
\rho_{*}\otimes\tau_{3}\otimes\cdots\otimes\tau_{n}.
\]
For this tensor-extension construction one still has $\Delta_{\widetilde{B}}\left(\rho_{*}\otimes\tau_{3}\otimes\cdots\otimes\tau_{n}\right)>0$,
and the defect denominator 
\[
m+\sum_{1\leq i<j\leq m}\phi^{\left(n\right)}_{ij}
\]
coincides with the bipartite case. Constructing genuinely multipartite
examples, in which all tensor factors enter nontrivially and both
the product threshold and the defect denominator can be evaluated
in closed form, would require a more involved optimization over product
states and is beyond the scope of this paper. 
\end{rem}

The information bound from \prettyref{thm:3-1} still depends on the
product threshold $\Gamma_{\mathrm{prod}}\left(B\right)$. The next
theorem shows that under a simple uniform $\ell^{2}$ bound on the
local expectation vectors this threshold is controlled by a purely
local quantity.
\begin{thm}
\label{thm:3-4} In the setting of \prettyref{prop:2-1}, assume that
$n\geq2$. Let $\mathcal{P}$ be the set of product states on $\mathcal{H}$
as in \prettyref{eq:2-2}. Suppose that for each $r\in\left\{ 1,\dots,n\right\} $
there exists a constant $C_{r}\geq0$ such that 
\begin{equation}
\sum^{m}_{i=1}\left|\mathrm{Tr}\left(\sigma^{\left(r\right)}a^{\left(r\right)}_{i}\right)\right|^{2}\leq C_{r}\label{eq:3-6}
\end{equation}
for every state $\sigma^{\left(r\right)}$ on $H_{r}$. Then the product
threshold in \prettyref{eq:2-3} satisfies 
\begin{equation}
\Gamma_{\mathrm{prod}}\left(B\right)\leq\prod^{n}_{r=1}C^{1/2}_{r}.\label{eq:3-7}
\end{equation}
\end{thm}

\begin{proof}
Let $\sigma=\sigma^{\left(1\right)}\otimes\cdots\otimes\sigma^{\left(n\right)}\in\mathcal{P}$
be a product state. For each $r\in\left\{ 1,\dots,n\right\} $ and
$i\in\left\{ 1,\dots,m\right\} $ set 
\[
x^{\left(r\right)}_{i}:=\mathrm{Tr}\left(\sigma^{\left(r\right)}a^{\left(r\right)}_{i}\right),\qquad x^{\left(r\right)}:=\left(x^{\left(r\right)}_{1},\dots,x^{\left(r\right)}_{m}\right)\in\mathbb{C}^{m}.
\]
Then 
\[
\mathrm{Tr}\left(\sigma B\right)=\sum^{m}_{i=1}\prod^{n}_{r=1}x^{\left(r\right)}_{i}.
\]

By \prettyref{eq:3-6}, each $x^{\left(r\right)}$ obeys 
\[
\Vert x^{\left(r\right)}\Vert^{2}_{2}=\sum^{m}_{i=1}|x^{\left(r\right)}_{i}|^{2}\leq C_{r},\qquad\text{so}\qquad\Vert x^{\left(r\right)}\Vert_{2}\leq C^{1/2}_{r}.
\]

We estimate the absolute value of the multilinear form. For each $i$,
\[
\left|\prod\nolimits^{n}_{r=1}x^{\left(r\right)}_{i}\right|=|x^{\left(1\right)}_{i}x^{\left(2\right)}_{i}|\prod\nolimits^{n}_{r=3}|x^{\left(r\right)}_{i}|\leq|x^{\left(1\right)}_{i}x^{\left(2\right)}_{i}|\prod\nolimits^{n}_{r=3}\Vert x^{\left(r\right)}\Vert_{2},
\]
since $|x^{\left(r\right)}_{i}|\leq\Vert x^{\left(r\right)}\Vert_{2}$
for each $r$. Summing over $i$ gives 
\[
\sum\nolimits^{m}_{i=1}\left|\prod\nolimits^{n}_{r=1}x^{\left(r\right)}_{i}\right|\leq\left(\prod\nolimits^{n}_{r=3}\Vert x^{\left(r\right)}\Vert_{2}\right)\sum\nolimits^{m}_{i=1}|x^{\left(1\right)}_{i}x^{\left(2\right)}_{i}|.
\]
By the Cauchy-Schwarz inequality, 
\[
\sum^{m}_{i=1}|x^{\left(1\right)}_{i}x^{\left(2\right)}_{i}|\leq\Vert x^{\left(1\right)}\Vert_{2}\Vert x^{\left(2\right)}\Vert_{2},
\]
so overall 
\[
\sum\nolimits^{m}_{i=1}\left|\prod\nolimits^{n}_{r=1}x^{\left(r\right)}_{i}\right|\leq\prod\nolimits^{n}_{r=1}\Vert x^{\left(r\right)}\Vert_{2}\leq\prod\nolimits^{n}_{r=1}C^{1/2}_{r}.
\]

It follows that 
\[
\left|\mathrm{Tr}\left(\sigma B\right)\right|\leq\prod^{n}_{r=1}C^{1/2}_{r}
\]
for every product state $\sigma\in\mathcal{P}$. In particular, 
\[
\Gamma_{\mathrm{prod}}\left(B\right)=\sup_{\sigma\in\mathcal{P}}\mathrm{Tr}\left(\sigma B\right)\leq\sup_{\sigma\in\mathcal{P}}\left|\mathrm{Tr}\left(\sigma B\right)\right|\leq\prod^{n}_{r=1}C^{1/2}_{r},
\]
which proves the theorem. 
\end{proof}

Combining the product threshold estimate with \prettyref{thm:3-1}
yields the explicit form of the correlation bound.
\begin{cor}
\label{cor:3-5} In the setting of \prettyref{thm:3-4}, let $\rho$
be a state on $\mathcal{H}$. Then 
\[
I_{\mathrm{tot}}\left(\rho\right)\geq\frac{1}{2}\frac{\left(\mathrm{Tr}\left(\rho B\right)-\prod^{n}_{r=1}C^{1/2}_{r}\right)^{2}_{+}}{m+\sum_{1\leq i<j\leq m}\phi^{\left(n\right)}_{ij}}.
\]
\end{cor}

\begin{proof}
By \prettyref{thm:3-4}, 
\[
\Gamma_{\mathrm{prod}}\left(B\right)\leq\prod^{n}_{r=1}C^{1/2}_{r}.
\]
Therefore 
\[
\Delta_{B}\left(\rho\right)=\left(\mathrm{Tr}\left(\rho B\right)-\Gamma_{\mathrm{prod}}\left(B\right)\right)_{+}\geq\left(\mathrm{Tr}\left(\rho B\right)-\prod^{n}_{r=1}C^{1/2}_{r}\right)_{+}.
\]
Applying \prettyref{thm:3-1}, we obtain 
\[
I_{\mathrm{tot}}\left(\rho\right)\geq\frac{1}{2}\frac{\Delta_{B}\left(\rho\right)^{2}}{m+\sum_{1\leq i<j\leq m}\phi^{\left(n\right)}_{ij}}\geq\frac{1}{2}\frac{\left(\mathrm{Tr}\left(\rho B\right)-\prod^{n}_{r=1}C^{1/2}_{r}\right)^{2}_{+}}{m+\sum_{1\leq i<j\leq m}\phi^{\left(n\right)}_{ij}},
\]
which proves the corollary. 
\end{proof}

\begin{example}
\label{exa:Pauli-site-bessel} We give a concrete family for which
the site condition \prettyref{eq:3-6} holds with a sharp, nontrivial
constant and for which the resulting product threshold and defect
denominator can be evaluated explicitly.

Assume that each local Hilbert space is a qubit, $H_{r}=\mathbb{C}^{2}$,
$r\in\left\{ 1,\dots,n\right\} $, and let 
\[
X=\begin{pmatrix}0 & 1\\
1 & 0
\end{pmatrix},\qquad Y=\begin{pmatrix}0 & -i\\
i & 0
\end{pmatrix},\qquad Z=\begin{pmatrix}1 & 0\\
0 & -1
\end{pmatrix}
\]
be the Pauli matrices. For each site $r$, set 
\[
a^{\left(r\right)}_{1}:=X,\qquad a^{\left(r\right)}_{2}:=Y,\qquad a^{\left(r\right)}_{3}:=Z.
\]
Thus $m=3$, each $a^{\left(r\right)}_{i}$ is a self-adjoint unitary,
and 
\[
B=\sum^{3}_{i=1}a^{\left(1\right)}_{i}\otimes\cdots\otimes a^{\left(n\right)}_{i}=X^{\otimes n}+Y^{\otimes n}+Z^{\otimes n}.
\]

We first verify \prettyref{eq:3-6}. Let $\sigma^{\left(r\right)}$
be a state on $H_{r}$. Writing $\sigma^{\left(r\right)}$ in Bloch
form, 
\[
\sigma^{\left(r\right)}=\frac{1}{2}\left(I+x_{r}X+y_{r}Y+z_{r}Z\right),
\]
with $x^{2}_{r}+y^{2}_{r}+z^{2}_{r}\leq1$, we have 
\[
\mathrm{Tr}\left(\sigma^{\left(r\right)}X\right)=x_{r},\qquad\mathrm{Tr}\left(\sigma^{\left(r\right)}Y\right)=y_{r},\qquad\mathrm{Tr}\left(\sigma^{\left(r\right)}Z\right)=z_{r}.
\]
Therefore 
\[
\sum^{3}_{i=1}\left|\mathrm{Tr}\left(\sigma^{\left(r\right)}a^{\left(r\right)}_{i}\right)\right|^{2}=x^{2}_{r}+y^{2}_{r}+z^{2}_{r}\leq1.
\]
Hence \prettyref{eq:3-6} holds at every site with the sharp constant
\[
C_{r}=1,\qquad r\in\left\{ 1,\dots,n\right\} .
\]

It follows from \prettyref{eq:3-7} that 
\[
\Gamma_{\mathrm{prod}}\left(B\right)\leq\prod^{n}_{r=1}C^{1/2}_{r}=1.
\]
On the other hand, if $\omega_{0}:=\left|0\right\rangle \left\langle 0\right|$
denotes the rank-one projection onto the $+1$ eigenspace of $Z$,
then for the product state $\sigma_{0}:=\omega^{\otimes n}_{0}$ we
have 
\[
\mathrm{Tr}\left(\sigma_{0}X^{\otimes n}\right)=0,\qquad\mathrm{Tr}\left(\sigma_{0}Y^{\otimes n}\right)=0,\qquad\mathrm{Tr}\left(\sigma_{0}Z^{\otimes n}\right)=1,
\]
and hence $\mathrm{Tr}\left(\sigma_{0}B\right)=1$. Therefore $\Gamma_{\mathrm{prod}}\left(B\right)=1$. 

We next compute the defect denominator from \prettyref{prop:2-1}.
For each pair $1\leq i<j\leq3$, the local pair $\left(a^{\left(r\right)}_{i},a^{\left(r\right)}_{j}\right)$
consists of two distinct Pauli matrices. Hence for every site $r$,
\[
\left\Vert \left\{ a^{\left(r\right)}_{i},a^{\left(r\right)}_{j}\right\} \right\Vert =0,\qquad\left\Vert \left[a^{\left(r\right)}_{i},a^{\left(r\right)}_{j}\right]\right\Vert =2.
\]
Thus in the defining sum for $\phi^{\left(n\right)}_{ij}$, every
term vanishes unless $S=\left\{ 1,\dots,n\right\} $. This subset
has even cardinality exactly when $n$ is even. Therefore 
\[
\phi^{\left(n\right)}_{ij}=\begin{cases}
0, & n\ \text{odd},\\
2^{1-n}\cdot2^{n}=2, & n\ \text{even},
\end{cases}\qquad1\leq i<j\leq3.
\]
Since there are $\binom{3}{2}=3$ such pairs, we obtain 
\[
m+\sum_{1\leq i<j\leq3}\phi^{\left(n\right)}_{ij}=\begin{cases}
3, & n\ \text{odd},\\
9, & n\ \text{even}.
\end{cases}
\]

Consequently, for every state $\rho$ on $H_{1}\otimes\cdots\otimes H_{n}$,
\[
I_{\mathrm{tot}}\left(\rho\right)\geq\begin{cases}
\frac{1}{6}\left(\mathrm{Tr}\left(\rho B\right)-1\right)^{2}_{+}, & n\ \text{odd},\\
\frac{1}{18}\left(\mathrm{Tr}\left(\rho B\right)-1\right)^{2}_{+}, & n\ \text{even},
\end{cases}
\]
by \prettyref{cor:3-5}.

Thus this Pauli configuration gives a concrete family in which the
site condition \prettyref{eq:3-6}, the product threshold, the parity-defect
denominator, and the resulting static correlation bounds are all explicit. 
\end{example}

\section{Correlation Decay Under Local Noise}\label{sec:4}

We next pass from the static estimate to a dynamical setting. The
preceding lower bound on total correlation converts any quantitative
decay of $I_{\mathrm{tot}}$ under a product quantum Markov semigroup
into a corresponding decay estimate for the excess $\Delta_{B}\left(\rho_{t}\right)$
over the product threshold. In particular, fixing a tolerance $\varepsilon>0$,
we obtain an explicit upper bound on the times $t$ for which $\Delta_{B}(\rho_{t})$
can still exceed $\varepsilon$; we refer to this as a survival time
bound for the excess above the product threshold. This yields a decay-to-threshold
estimate in terms of the semigroup rate and the same parity-defect
denominator that appears in \prettyref{thm:3-1}, placing our result
in the general setting of quantitative entropy decay and functional
inequalities for quantum Markov semigroups; see for example \cite{MR4282395,MR4496596,MR4703456,MR4704528}.
\begin{prop}
\label{prop:4-1} In the setting of \prettyref{prop:2-1} and \prettyref{thm:3-1},
let $\left(\mathcal{T}_{t}\right)_{t\geq0}$ be a one parameter semigroup
of completely positive trace preserving maps on $\mathcal{H}$, and
write 
\[
\rho_{t}:=\mathcal{T}_{t}\left(\rho_{0}\right)
\]
for the evolution of an initial state $\rho_{0}$. Assume that there
exists a constant $\lambda>0$ such that for every initial state $\rho_{0}$
on $\mathcal{H}$ and every $t\geq0$, the total correlation 
\[
I_{\mathrm{tot}}\left(\rho\right):=D\left(\rho\Vert\rho_{1}\otimes\cdots\otimes\rho_{n}\right)
\]
satisfies 
\begin{equation}
I_{\mathrm{tot}}\left(\rho_{t}\right)\leq\mathrm{e}^{-2\lambda t}I_{\mathrm{tot}}\left(\rho_{0}\right).\label{eq:4-1}
\end{equation}
Let $B$ be the multipartite observable from \prettyref{prop:2-1},
and let $\Gamma_{\mathrm{prod}}\left(B\right)$, $\Delta_{B}\left(\rho\right)$
be the associated product threshold and excess from \prettyref{thm:3-1}.
Then the following hold for every initial state $\rho_{0}$. 
\begin{enumerate}
\item For every $t\geq0$, 
\begin{equation}
\Delta_{B}\left(\rho_{t}\right)\leq\mathrm{e}^{-\lambda t}\left(2I_{\mathrm{tot}}\left(\rho_{0}\right)\right)^{1/2}\left(m+\sum_{1\leq i<j\leq m}\phi^{\left(n\right)}_{ij}\right)^{1/2}.\label{eq:4-2}
\end{equation}
In particular, let $\varepsilon>0$ and assume $\Delta_{B}\left(\rho_{0}\right)>0$.
Then every time $t\geq0$ satisfying 
\begin{equation}
t\geq\frac{1}{\lambda}\log\left(\frac{\left(2I_{\mathrm{tot}}\left(\rho_{0}\right)\right)^{1/2}\left(m+\sum_{1\leq i<j\leq m}\phi^{\left(n\right)}_{ij}\right)^{1/2}}{\varepsilon}\right)\label{eq:4-3}
\end{equation}
obeys $\Delta_{B}\left(\rho_{t}\right)\leq\varepsilon$. 
\item The squared excess is integrable in time and satisfies 
\begin{equation}
\int^{\infty}_{0}\Delta_{B}\left(\rho_{t}\right)^{2}dt\leq\frac{1}{\lambda}I_{\mathrm{tot}}\left(\rho_{0}\right)\left(m+\sum_{1\leq i<j\leq m}\phi^{\left(n\right)}_{ij}\right).\label{eq:4-4}
\end{equation}
That is, the total time integrated squared excess above the product
threshold that is detectable by $B$ along the orbit $\left(\rho_{t}\right)_{t\geq0}$
is controlled by the initial total correlation and the parity defect
denominator. 
\end{enumerate}
\end{prop}

\begin{proof}
Fix an initial state $\rho_{0}$ on $\mathcal{H}$ and set $\rho_{t}:=\mathcal{T}_{t}\left(\rho_{0}\right)$.
For brevity, write 
\[
M:=m+\sum_{1\leq i<j\leq m}\phi^{\left(n\right)}_{ij}.
\]

We first prove (1). By \prettyref{thm:3-1}, applied to $\rho_{t}$,
we have 
\[
I_{\mathrm{tot}}\left(\rho_{t}\right)\geq\frac{1}{2}\frac{\Delta_{B}\left(\rho_{t}\right)^{2}}{M}
\]
for every $t\geq0$. Equivalently, 
\[
\Delta_{B}\left(\rho_{t}\right)^{2}\leq2MI_{\mathrm{tot}}\left(\rho_{t}\right).
\]
Using the decay assumption \prettyref{eq:4-1} gives 
\[
\Delta_{B}\left(\rho_{t}\right)^{2}\leq2M\mathrm{e}^{-2\lambda t}I_{\mathrm{tot}}\left(\rho_{0}\right),
\]
and taking square roots yields 
\[
\Delta_{B}\left(\rho_{t}\right)\leq\mathrm{e}^{-\lambda t}\left(2I_{\mathrm{tot}}\left(\rho_{0}\right)\right)^{1/2}M^{1/2},
\]
which is \prettyref{eq:4-2}.

For the survival time statement, let $\varepsilon>0$ and assume $\Delta_{B}\left(\rho_{0}\right)>0$.
From \prettyref{eq:4-2} we have 
\[
\Delta_{B}\left(\rho_{t}\right)\leq\mathrm{e}^{-\lambda t}\left(2I_{\mathrm{tot}}\left(\rho_{0}\right)\right)^{1/2}M^{1/2}.
\]
Thus every time $t\geq0$ satisfying 
\[
\mathrm{e}^{-\lambda t}\left(2I_{\mathrm{tot}}\left(\rho_{0}\right)\right)^{1/2}M^{1/2}\leq\varepsilon
\]
obeys $\Delta_{B}\left(\rho_{t}\right)\leq\varepsilon$. Solving this
inequality for $t$ gives 
\[
t\geq\frac{1}{\lambda}\log\left(\frac{\left(2I_{\mathrm{tot}}\left(\rho_{0}\right)\right)^{1/2}M^{1/2}}{\varepsilon}\right),
\]
which is \prettyref{eq:4-3}.

We now prove (2). Starting again from 
\[
\Delta_{B}\left(\rho_{t}\right)^{2}\leq2MI_{\mathrm{tot}}\left(\rho_{t}\right),
\]
and using \prettyref{eq:4-1}, we obtain 
\[
\Delta_{B}\left(\rho_{t}\right)^{2}\leq2M\mathrm{e}^{-2\lambda t}I_{\mathrm{tot}}\left(\rho_{0}\right),\qquad t\geq0.
\]
Integrating both sides over $t\in\left[0,\infty\right)$ yields 
\[
\int^{\infty}_{0}\Delta_{B}\left(\rho_{t}\right)^{2}dt\leq2MI_{\mathrm{tot}}\left(\rho_{0}\right)\int^{\infty}_{0}\mathrm{e}^{-2\lambda t}dt.
\]
Since $\int^{\infty}_{0}\mathrm{e}^{-2\lambda t}dt=\frac{1}{2\lambda}$,
we obtain \prettyref{eq:4-4} from this. 
\end{proof}

The preceding proposition becomes completely explicit for local depolarizing
noise. In this case the Heisenberg evolution acts diagonally on centered
observables, so the multipartite observable $B$ decays by an exact
scalar factor. This gives a concrete model in which the excess $\Delta_{B}\left(\rho_{t}\right)$
and the associated lower bound on total correlation can be written
in closed form.
\begin{example}
Assume in addition that each $H_{r}$ is finite dimensional, and let
\[
\tau_{r}\left(X\right):=\frac{1}{\dim H_{r}}\mathrm{Tr}\left(X\right),\qquad X\in B\left(H_{r}\right),
\]
be the normalized trace on $B\left(H_{r}\right)$. For each $r\in\left\{ 1,\dots,n\right\} $,
let $\left(\mathcal{T}^{\left(r\right)}_{t}\right)_{t\geq0}$ be the
depolarizing semigroup whose Heisenberg action is 
\[
\mathcal{T}^{\left(r\right)*}_{t}\left(X\right)=\mathrm{e}^{-t}X+\left(1-\mathrm{e}^{-t}\right)\tau_{r}\left(X\right)I_{H_{r}},\qquad X\in B\left(H_{r}\right).
\]
Set 
\[
\mathcal{T}_{t}:=\mathcal{T}^{\left(1\right)}_{t}\otimes\cdots\otimes\mathcal{T}^{\left(n\right)}_{t}
\]
on states over 
\[
\mathcal{H}=H_{1}\otimes\cdots\otimes H_{n}.
\]
Let 
\[
B=\sum^{m}_{i=1}a^{\left(1\right)}_{i}\otimes\cdots\otimes a^{\left(n\right)}_{i}
\]
be as in \prettyref{prop:2-1}, and assume moreover that 
\[
\tau_{r}\left(a^{\left(r\right)}_{i}\right)=0\qquad\text{for all }r\in\left\{ 1,\dots,n\right\} \text{ and }i\in\left\{ 1,\dots,m\right\} .
\]
Let $\Gamma_{\mathrm{prod}}\left(B\right)$ and $\Delta_{B}$ be as
in \prettyref{thm:3-1}. Then the following hold. 
\begin{enumerate}
\item The observable $B$ evolves by the exact formula 
\[
\mathcal{T}^{*}_{t}\left(B\right)=\mathrm{e}^{-nt}B.
\]
Consequently, for every initial state $\rho_{0}$ and every $t\geq0$,
if $\rho_{t}:=\mathcal{T}_{t}\left(\rho_{0}\right)$, then 
\[
\mathrm{Tr}\left(\rho_{t}B\right)=\mathrm{e}^{-nt}\mathrm{Tr}\left(\rho_{0}B\right).
\]
\item For every initial state $\rho_{0}$ and every $t\geq0$, 
\[
\Delta_{B}\left(\rho_{t}\right)=\left(\mathrm{Tr}\left(\rho_{t}B\right)-\Gamma_{\mathrm{prod}}\left(B\right)\right)_{+}=\left(\mathrm{e}^{-nt}\mathrm{Tr}\left(\rho_{0}B\right)-\Gamma_{\mathrm{prod}}\left(B\right)\right)_{+}.
\]
If $\Gamma_{\mathrm{prod}}\left(B\right)>0$ and $\mathrm{Tr}\left(\rho_{0}B\right)>\Gamma_{\mathrm{prod}}\left(B\right)$,
then 
\[
\Delta_{B}\left(\rho_{t}\right)>0\qquad\text{for all }0\leq t<\frac{1}{n}\log\left(\frac{\mathrm{Tr}\left(\rho_{0}B\right)}{\Gamma_{\mathrm{prod}}\left(B\right)}\right).
\]
If $\Gamma_{\mathrm{prod}}\left(B\right)\leq0$ and $\mathrm{Tr}\left(\rho_{0}B\right)>\Gamma_{\mathrm{prod}}\left(B\right)$,
then 
\[
\Delta_{B}\left(\rho_{t}\right)>0\qquad\text{for all }t\geq0.
\]
\item For every initial state $\rho_{0}$ and every $t\geq0$, 
\[
I_{\mathrm{tot}}\left(\rho_{t}\right)\geq\frac{1}{2}\frac{\left(\mathrm{e}^{-nt}\mathrm{Tr}\left(\rho_{0}B\right)-\Gamma_{\mathrm{prod}}\left(B\right)\right)^{2}_{+}}{m+\sum_{1\leq i<j\leq m}\phi^{\left(n\right)}_{ij}},
\]
where $I_{\mathrm{tot}}\left(\rho_{t}\right)$ is the total correlation
from \prettyref{thm:3-1}. 
\end{enumerate}
\end{example}

\begin{proof}
For each $r$ and $i$, the centering assumption gives 
\[
\mathcal{T}^{\left(r\right)*}_{t}\left(a^{\left(r\right)}_{i}\right)=\mathrm{e}^{-t}a^{\left(r\right)}_{i}
\]
because $\tau_{r}\left(a^{\left(r\right)}_{i}\right)=0$. Therefore
\[
\mathcal{T}^{*}_{t}\left(a^{\left(1\right)}_{i}\otimes\cdots\otimes a^{\left(n\right)}_{i}\right)=\mathrm{e}^{-nt}a^{\left(1\right)}_{i}\otimes\cdots\otimes a^{\left(n\right)}_{i},
\]
and summing over $i$ yields 
\[
\mathcal{T}^{*}_{t}\left(B\right)=\mathrm{e}^{-nt}B.
\]
Hence, for every initial state $\rho_{0}$, 
\[
\mathrm{Tr}\left(\rho_{t}B\right)=\mathrm{Tr}\left(\rho_{0}\mathcal{T}^{*}_{t}\left(B\right)\right)=\mathrm{e}^{-nt}\mathrm{Tr}\left(\rho_{0}B\right),
\]
which proves (1).

Statement (2) follows immediately from the definition 
\[
\Delta_{B}\left(\rho_{t}\right)=\left(\mathrm{Tr}\left(\rho_{t}B\right)-\Gamma_{\mathrm{prod}}\left(B\right)\right)_{+}
\]
and the formula from (1). If $\Gamma_{\mathrm{prod}}\left(B\right)>0$,
then 
\[
\mathrm{e}^{-nt}\mathrm{Tr}\left(\rho_{0}B\right)>\Gamma_{\mathrm{prod}}\left(B\right)
\]
is equivalent to 
\[
t<\frac{1}{n}\log\left(\frac{\mathrm{Tr}\left(\rho_{0}B\right)}{\Gamma_{\mathrm{prod}}\left(B\right)}\right),
\]
which gives the stated positivity interval. If $\Gamma_{\mathrm{prod}}\left(B\right)\leq0$
and $\mathrm{Tr}\left(\rho_{0}B\right)>\Gamma_{\mathrm{prod}}\left(B\right)$,
then for every $t\geq0$ we have 
\[
\mathrm{e}^{-nt}\mathrm{Tr}\left(\rho_{0}B\right)\geq\min\left\{ \mathrm{Tr}\left(\rho_{0}B\right),0\right\} >\Gamma_{\mathrm{prod}}\left(B\right),
\]
so $\Delta_{B}\left(\rho_{t}\right)>0$ for all $t\geq0$.

Finally, (3) follows by substituting the identity from (2) into the
lower bound from \prettyref{thm:3-1}. 
\end{proof}

\section{Concluding Remarks}

The results of \prettyref{sec:2} show that the multipartite sum 
\[
B=\sum^{m}_{i=1}a^{\left(1\right)}_{i}\otimes\cdots\otimes a^{\left(n\right)}_{i}
\]
carries a built-in parity structure. After expanding $B^{2}$, the
mixed terms do not generate an arbitrary family of tensor expressions:
the odd parity contributions cancel, and the surviving even parity
sector produces the defect weights $\phi^{\left(n\right)}_{ij}$.
These defect weights give a direct norm bound for $B$ in terms of
commutator and anticommutator data, and provide the denominator that
controls the correlation estimates throughout the paper. The tripartite
Pauli example in \prettyref{sec:2} shows that this norm bound can
already be sharp for genuinely multipartite configurations.

The same quantity links operator control to total correlation. Once
the product threshold is bounded, any positive excess of the value
of $B$ above that threshold forces both separation from the product
class in trace norm and a quantitative lower bound on $I_{\mathrm{tot}}(\rho)$.
Under the uniform $\ell^{2}$ assumption on the local expectation
vectors, this lower bound becomes fully explicit, with the product
threshold in the numerator and the parity defect denominator in the
denominator. The CHSH configuration and the Pauli-site family in \prettyref{sec:3}
illustrate this in concrete bipartite and multipartite settings where
the threshold, the defect denominator, and the resulting lower bounds
on total correlation can all be computed in closed form.

The dynamical results in \prettyref{sec:4} show that any quantitative
decay of $I_{\mathrm{tot}}(\rho_{t})$ along a quantum Markov semigroup
can be converted into a pointwise decay bound and a time-integrated
bound for the excess $\Delta_{B}(\rho_{t})$ above the product threshold,
with the same parity defect denominator appearing as the prefactor.
In the case of local depolarizing noise, the observable $B$ is an
eigenvector for the Heisenberg evolution, so $\Delta_{B}(\rho_{t})$
and the corresponding lower bound on total correlation can be tracked
explicitly along the flow.

A natural next step would be to identify operator families for which
the constants $C_{r}$ admit a more intrinsic description, for example
through orthogonality relations, frame-type conditions, or Clifford-type
geometry at each site, and to understand how such structure is reflected
in the defect weights $\phi^{\left(n\right)}_{ij}$. It would also
be interesting to understand whether other local semigroups lead to
explicit evolution formulas for $B$ or for the associated defect
data, and to relate bounds on $\Delta_{B}(\rho_{t})$ to functional
inequalities beyond the exponential decay of $I_{\mathrm{tot}}$.
Some recent directions related to entropy decay, the dynamics of such
observables, and geometric features of quantum channels may be found
in \cite{MR4703456,MR4704528,zbMATH07883336,MR4469657}.

\bibliographystyle{amsalpha}
\bibliography{ref}

\end{document}